\documentclass[twoside, 11pt]{article}
\usepackage{a4}
\usepackage{amsmath, amssymb}
\usepackage[utf8x]{inputenc}
\usepackage[american]{babel}
\usepackage{pifont}
\usepackage{theorem}
\newtheorem{dfn}{Definition}

\newtheorem{theo}[dfn]{Theorem}

\newtheorem{rem}[dfn]{Remark}
\newtheorem{lem}[dfn]{Lemma}
\newtheorem{fact}[dfn]{Fact}
\newtheorem{ex}[dfn]{Example}

\renewcommand{\phi}{\varphi}
\usepackage[pdftitle={Hypertree-depth},pdfauthor={Isolde Adler}]{hyperref}
\usepackage{tikz}
\usetikzlibrary{calc} 
\usetikzlibrary{decorations} 
\usetikzlibrary{shadows}
\usetikzlibrary{snakes} 
\usetikzlibrary{shapes}

\newcommand{\lrw}{\operatorname{lrw}}
\newcommand{\rw}{\operatorname{rw}}
\newcommand{\OO}{\mathcal O}
\newcommand{\CC}{\mathcal C}
\newcommand{\LL}{\mathcal L}
\newcommand{\BB}{\mathcal B}

\newcommand{\vm}{\preccurlyeq_v}
\newcommand{\cutrk}{\operatorname{cutrk}}
\newcommand{\rank}{\operatorname{rank}}
\newcommand{\dist}{\operatorname{dist}}
\newcommand{\bigmid}{\;\big|\;}
\newcommand{\dotcup}{\,\dot\cup\,}
\newcommand{\qed}{\hspace*{\fill}$\square$}

\def\vertexnodes{\tikzstyle{every node}=[fill=black,circle, inner sep=2.2pt]}

\title{Obstructions for linear rankwidth at most $1$}
\author{Isolde Adler\thanks{Institute for Computer Science, Goethe University Frankfurt, Germany},
\def\myfoot{\thefootnote}
Arthur M.\ Farley\thanks{Computer and Information Science Department, University of Oregon,
Eugene, Oregon, USA}, 
Andrzej Proskurowski$^\myfoot$}

\begin{document}

\maketitle
\begin{abstract}
	We provide a characterization of graphs of linear rankwidth at most~$1$ by minimal excluded
	vertex-minors.
\end{abstract}

\section{Introduction}

The definition and study  of various
{\em width parameters} of graphs has influenced research on
structural characterizations and exploring complexity and
algorithmic properties of graph classes with bounded width. 
One of the first such parameters was {\em bandwidth}, discussed for instance in
papers by Monien and Sudborough~\cite{MonienS85}, Chinn et al~\cite{Chinn82},
Assman et al~\cite{Assman81}. The first modern width parameter was 
{\em treewidth} defined by Robertson and Seymour~\cite{RobertsonS86}, 
opening the floodgates for various
graph decomposition schemes that define other width parameters. These
parameters have strong impact on complexity of many discrete optimization
problems. 

{\em Rankwidth} was first defined by Oum and Seymour~\cite{OumS06} with the
goal of efficient approximation of the cliquewidth of a graph. 
Oum showed that the rankwidth cannot increase when taking {\em vertex minors}~\cite{Oum05},
and he further
investigated the problem of {\em obstruction set characterization} of graphs
with bounded rankwidth. He proved that for given rankwidth $k\ge 0$, the
obstructions (defined as minimal excluded vertex-minors) have bounded size. In
the same paper, he showed that a graph has rankwidth at most $1$ if and only if
it is distance-hereditary. It then follows from results in~\cite{BandeltM86} that the 
obstruction set for graphs of rankwidth at most $1$ simply consists of the $5$-cycle $C_5$.
In~\cite{Bouchet94}, Bouchet determined the obstruction set characterizations for circle graphs. 

The main theorem in this paper is a characterization of 
the class of all graphs of linear rankwidth at most $1$ by three excluded
vertex-minors.
\begin{theo}\label{theo:main}
	Any graph $G$ has linear rankwidth at most $1$ if and only if $G$ contains none
	of the three graphs depicted in Figure~\ref{fig:obstructions} as a vertex-minor.
\end{theo}

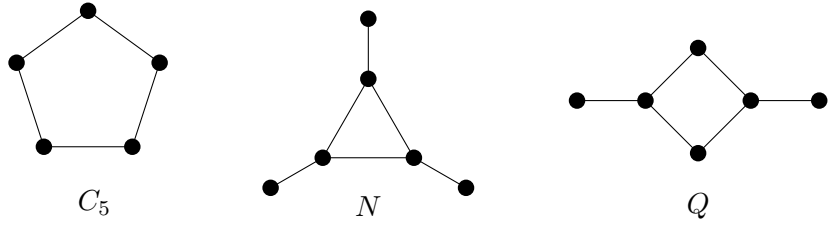
\begin{figure}[h]
\begin{center}	
\begin{tikzpicture}
	\begin{scope}[rotate=18]
	\vertexnodes 
	\coordinate(name) at (-0.4,-1.5);
	\path (0,0) coordinate (origin); 
	\path (0:1cm) coordinate (P0);
	\path (1*72:1cm) coordinate (P1);
	\path (2*72:1cm) coordinate (P2);
	\path (3*72:1cm) coordinate (P3);
	\path (4*72:1cm) coordinate (P4);

	\draw (P0) node{} -- (P1) node{} -- (P2) node{} -- (P3) node{} -- (P4) node{} -- cycle; 
	\draw (name) node[fill=none]{$C_5$};
	\end{scope}
\end{tikzpicture}
\hspace{1cm}
\begin{tikzpicture}
	\vertexnodes 
	\coordinate(name) at (0,-1);
	\coordinate(top) at (90:0.7cm);
	\coordinate(left) at (210:0.7cm);
	\coordinate(right) at (-30:0.7cm);
	\coordinate(extop) at (90:1.5cm);
	\coordinate(exleft) at (210:1.5cm);
	\coordinate(exright) at (-30:1.5cm);
	\draw  (top) node{} -- (left) node{} -- (right) node{} -- cycle;
	\draw (top) -- (extop) node{};
	\draw (left) -- (exleft) node{};
	\draw (right) -- (exright) node{};
	\draw (name) node[fill=none]{$N$};

\end{tikzpicture}
\hspace{1cm}
\begin{tikzpicture}
	\vertexnodes 
\begin{scope}[scale=.7]
	\coordinate(name) at (2,-2);
	\coordinate(top2) at (2,1);
	\coordinate(left2) at (1,0);
	\coordinate(right2) at (3,0);
	\coordinate(bot2) at (2,-1);
	\coordinate(a) at (-0.3,0);
	\coordinate(b) at (4.3,0);
	\draw  (top2) node{} -- (left2) node{} -- (bot2) node{} --(right2) node{}-- cycle;
	\draw (a) node{} -- (left2);
	\draw (b) node{} -- (right2);
	\draw (name) node[fill=none]{$Q$};
\end{scope}

\end{tikzpicture}

\end{center}
\caption{The three obstructions for linear rankwidth at most $1$: The $5$-cycle $C_5$, the net graph $N$, and
the half-cube $Q$.}\label{fig:obstructions}
\end{figure}

It is known that for every fixed integer $k> 0$, the set $\OO_k$ of
graphs that are minimal excluded vertex-minors for linear rankwidth at most $k$ is \emph{finite}~\cite{Oum08}.
Until now, no such set $\OO_k$ was explicitly known. In this paper we determine $\OO_1$.

\section{Preliminaries}
For a set $A$ we denote the power set of $A$ by $2^A$. 
For two sets $A$ and $B$ let $A\Delta B:=(A\setminus B)\cup (B\setminus A)$ denote the
\emph{symmetric difference} of $A$ and $B$.
For an integer $n> 0$
we let $[n]:=\{1,\ldots,n\}$. 
Let $\bar v=v_1,\ldots,v_n$ be an ordered tuple. 
We say that a subtuple $\bar u$ of $\bar v$ is an \emph{interval} of
$\bar v$, if $\bar v$ can be written as
$\bar v=\bar x,\bar u,\bar y$ for some (possibly empty) tuples
$\bar x$ and $\bar y$.

Graphs are finite, simple and undirected.
We denote the set of vertices of $G$ by $V(G)$ and the set of
edges of $G$ by $E(G)$, and every edge $e\in E(G)$ is 
a two-element subset of $V(G)$.
Let $G$ be a graph. 
For a vertex $v$ we let $N_G(v):=\big\{u\in V(G)\bigmid \{u,v\}\in E(G)\big\}$ be the 
\emph{neighborhood} of $v$ in $G$. The \emph{degree} of $v\in V(G)$ is
$\deg_G(v):=\left|N_G(v)\right|$.
For two graph $G$ and $H$, the \emph{intersection} of $G$ and $H$ is the graph 
$G\cap H$ with $V(G\cap H):=V(G)\cap V(H)$ and $E(G\cap H)=E(G)\cap E(H)$.
The \emph{union} of $G$ and $H$ is the graph 
$G\cup H$ with $V(G\cap H):=V(G)\cup V(H)$ and $E(G\cup H)=E(G)\cup E(H)$.
A graph $H$ is a \emph{subgraph} of $G$, if $V(H)\subseteq V(G)$ and $E(H)\subseteq E(G)$.
For a subset $X\subseteq V(G)$, let $G[X]$ be the subgraph of $G$ \emph{induced}
by $X$, i.e.\ $V\big(G[X]\big)=X$ and $E\big(G[X]\big):=\{e\in E(G)\mid e\subseteq X\}$.
A graph $H$ is an \emph{induced subgraph} of $G$, if $H=G[X]$ for some
subset $X\subseteq V(G)$.
For a subset $Y\subseteq V(G)$ we let $G\setminus Y:=G[V(G)\setminus Y]$. If $Y=\{y\}$
is a singleton set, then we write $G\setminus y$ instead of $G\setminus \{y\}$.
We say that a vertex $v\in V(G)$ is a \emph{cut-vertex}, if $G\setminus v$ has
more connected components than $G$.
A graph $G$ is \emph{connected}, if $G\neq\emptyset$ and any two vertices of $G$ are connected by a path.
A subset $X\subseteq V(G)$ is \emph{connected}, if $G[X]$ is connected.
A graph $G$ that is not connected is said to be \emph{disconnected}. 
A \emph{connected component} of $G$ is a maximal connected subgraph of $G$. 
For an integer $k\geq 0$, 
a graph $G$ is $k$-connected if $G$ cannot be disconnected by removing fewer than $k$ vertices. 
(Hence in particular, a $k$-connected graph has at least $k$ vertices.)

The \emph{length} of a path is the number of its edges.
The \emph{distance} between two vertices $u$ and $v$ of $G$, denoted by $\dist_G(u,v)$, is the minimum
length of a path in $G$ connecting $u$ and $v$ (or infinity, if no such path exists).
A \emph{tree} $T$ is an acyclic connected graph. A \emph{leaf}
of a tree $T$ is a vertex of degree one in $T$. 
We denote the set of leaves of $T$ by $L(T)$.
A vertex in $V(T)\setminus L(T)$ is an \emph{internal vertex}.
For an integer $n\geq 3$ we let $C_n$ denote the cycle with $n$ vertices.
A \emph{complete bipartite} graph is a graph $G$ with a partition
$V(G)=X\dotcup Y$ such that $E(G)=\big\{\{x,y\}\mid x\in X\text{ and }y\in Y\big\}$.

\paragraph{Linear rankwidth}
For defining linear rankwidth, we introduce some notation.
Let $M(G)$ denote the adjacency matrix of a graph $G$, I.e.\ $M(G)$ is 
the $V(G)\times V(G)$ matrix where the columns and the rows are 
indexed by the vertices of $G$, and $M(G)$ has entries in $\{0,1\}$, where 
an entry is $1$ if and only if the corresponding row vertex is incident to the
corresponding column vertex.
For an $A\times B$ matrix $M$ and subsets $X\subseteq A$ and $Y\subseteq B$
we let $M[X,Y]$ denote the $X\times Y$ submatrix $(m_{i,j})_{i\in X,j\in Y}$
of $M$.

The \emph{cutrank function} of a graph $G$ is defined by $\cutrk_G\colon 2^{V(G)}\to\mathbb N$
given by \[\cutrk_G(X):=\rank\big(M(G)[X,V(G)\setminus X]\big),\]
where $\rank$ is the rank function over GF$[2]$.

A tree is \emph{cubic}, if it has at least two vertices and every internal vertex
has degree $3$. A \emph{rank decomposition} of a graph $G$ is a pair $(T,\lambda)$,
where $T$ is a cubic tree and $\lambda\colon L(T)\to V(G)$ is a bijection.
For every edge $e\in E(T)$ the two connected components  of $T\setminus e$ induce
a partition $(X_e,Y_e)$ of $L(T)$. The \emph{width} of $e$ is
defined as $\cutrk_G(\lambda(X_e))$. The \emph{width} of a rank decomposition
$(T,\lambda)$ is the maximum width over all edges of $T$. The \emph{rankwidth}
of $G$ is defined as 
\[\rw(G):=\min\{\text{width of }(T,\lambda)\mid  (T,\lambda)\text{ rank decomposition of }G\}.\]
(If $\left|V(G)\right|\leq 1$, then $G$ has no rank decomposition and we let $\rw(G):=0$.)

A \emph{caterpillar} is a tree $T$ that contains a path such that every vertex of $T$
has distance at most $1$ to some path vertex. 
A \emph{linear rank decomposition} of a graph $G$
is a rank decomposition $(T,\lambda)$ of $G$, where $T$ is a 
caterpillar. The \emph{linear rankwidth}
of $G$ is defined as
\[\lrw(G):=\min\{\text{width of }(T,\lambda)\mid  (T,\lambda)\text{ linear rank decomposition of }G\}.\]
(Again, if $\left|V(G)\right|\leq 1$, then $G$ 
has no linear rank decomposition and we let $\lrw(G):=0$.)

For example, it is easy to verify that cliques, caterpillars and complete bipartite 
graphs have linear rankwidth at most $1$, and that the disjoint union 
$G\dotcup H$ of two graphs $G$ and $H$ satisfies 
$\lrw(G\dotcup H)=\max\{\lrw(G),\lrw(H)\}$.

\begin{ex}\label{ex:seefive}
The cycle $C_5$ satisfies $\lrw(C_5)=2$: In any linear rank decomposition
$(T,\lambda)$ of $C_5$ every edge in $E(T)$ between two internal vertices of $T$
has width $2$, and every edge in $E(T)$ containing a leaf of $T$ has width $1$.
\end{ex}

\begin{rem}\label{rem:four-verts}
	All graphs on four vertices have linear rankwidth at most $1$.
\end{rem}

\paragraph{Vertex-minors, obstructions and distance-hereditary graphs}
Let $G$ be a graph and let $v\in V(G)$. The graph obtained from $G$ by a
\emph{local complementation at} $v$ is the graph $G*v$ with $V(G*v):=V(G)$
and $E(G*v):= E(G)\Delta \big\{\{x,y\}\subseteq N_G(v)\bigmid x\neq y\big\}$.
We say that two graphs $G$ and $H$ are \emph{locally equivalent}, $G\sim H$, if $H$ can be
obtained from $G$ by a sequence of local complementations. Note that this
is indeed an equivalence relation. Figure~\ref{fig:C5equiv} shows all graphs that are
locally equivalent to $C_5$ (up to isomorphism).

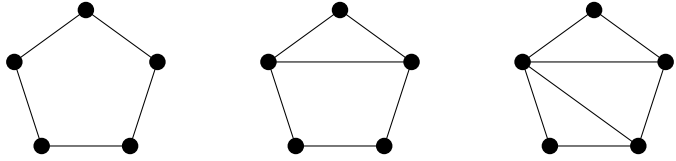
\begin{figure}[h]
\begin{center}	
\begin{tikzpicture}
	\begin{scope}[rotate=18]
	\vertexnodes 
	\coordinate(name) at (0,-1.5);
	\path (0,0) coordinate (origin); 
	\path (0:1cm) coordinate (P0);
	\path (1*72:1cm) coordinate (P1);
	\path (2*72:1cm) coordinate (P2);
	\path (3*72:1cm) coordinate (P3);
	\path (4*72:1cm) coordinate (P4);

	\draw (P0) node{} -- (P1) node{} -- (P2) node{} -- (P3) node{} -- (P4) node{} -- cycle; 
\end{scope}
\end{tikzpicture}
\hspace{1cm}
\begin{tikzpicture}
	\begin{scope}[rotate=18]
	\vertexnodes 
	\path (0,0) coordinate (origin); 
	\path (0:1cm) coordinate (P0);
	\path (1*72:1cm) coordinate (P1);
	\path (2*72:1cm) coordinate (P2);
	\path (3*72:1cm) coordinate (P3);
	\path (4*72:1cm) coordinate (P4);

	\draw (P0) node{} -- (P1) node{} -- (P2) node{} -- (P3) node{} -- (P4) node{} -- cycle; 
	\draw (P0) -- (P2);
	\end{scope}
\end{tikzpicture}
\hspace{1cm}
\begin{tikzpicture}
	\begin{scope}[rotate=18]
	\vertexnodes 
	\coordinate(name) at (0,-1.5);
	\path (0,0) coordinate (origin); 
	\path (0:1cm) coordinate (P0);
	\path (1*72:1cm) coordinate (P1);
	\path (2*72:1cm) coordinate (P2);
	\path (3*72:1cm) coordinate (P3);
	\path (4*72:1cm) coordinate (P4);

	\draw (P0) node{} -- (P1) node{} -- (P2) node{} -- (P3) node{} -- (P4) node{} -- cycle; 
	\draw (P0)  -- (P2) -- (P4);
	\end{scope}
\end{tikzpicture}

\end{center}
\caption{The three graphs that are locally equivalent to $C_5$.}\label{fig:C5equiv}
\end{figure}
A graph $H$ is a \emph{vertex-minor} of a graph $G$, denoted by $H\vm G$, if $H$ can be obtained from
$G$ by a sequence of local complementations and vertex deletions.

In particular, every induced subgraph of $G$ is a vertex-minor of $G$.
For a fixed non-negative integer $k\in\mathbb N$, the
class of all graphs of rankwidth at most $k$ is
closed under taking vertex-minors~\cite{Oum05}.
The following Lemma lists some basic observations
on linear rankwidth that are not hard to verify (cf.~\cite[Prop.~2.6]{Oum05}).

\begin{lem}\label{lem:vm-closed}
Let $G$ be a graph and let $v\in V(G)$.
\begin{enumerate}
	\item Every $X\subseteq V(G)$ satisfies $\cutrk_G(X)=\cutrk_{G*v}(X)$,
	\item $\lrw(G)=\lrw(G*v)$,
	\item $\lrw(G\setminus v)\leq\lrw(G)$,
	\item for fixed $k\in \mathbb N$,
		the class of all graphs of linear rankwidth at most $k$ is
		closed under taking vertex-minors.
\end{enumerate}
\end{lem}

A graph $G$ is \emph{distance-hereditary}, if for every induced connected subgraph $H\subseteq G$
and every pair of vertices $u,v\in V(H)$ we have $\dist_H(u,v)=\dist_G(u,v)$.

\begin{fact}[\cite{BandeltM86,Oum05}]\label{fact:seefive}
	For every graph $G$ the following are equivalent:
\begin{enumerate}
		\item $G$ is distance-hereditary,
		\item $\rw(G)\leq 1$, and
		\item $G$ does not contain
	$C_5$ as a vertex-minor.
\end{enumerate}
\end{fact}

Let $\CC$ be a class of graphs that is \emph{closed under taking vertex-minors}, i.e.\ all
graphs $G$ satisfy: if
$G\in\CC$ and $H\vm G$, then $H\in\CC$. 
We say that a graph $G$ is an \emph{obstruction}
for $\CC$, if every graph $H$ with $G\sim H$ satisfies 
\begin{itemize}
	\item $H\notin\CC$, and 
	\item for every $v\in V(H)$ the graph $H\setminus v$ is in $\CC$.
\end{itemize}

A set $\OO$ of graphs is an \emph{obstruction set} for $\CC$, if
$\OO$ is a set of pairwise locally non-equivalent obstructions for $\CC$, 
such that for every graph $G$, $G\in\CC$ if and only if $H\not\vm G$ for all $H\in\OO$.
For example, $\{C_5\}$ is an obstruction set for the class of distance-hereditary graphs,
and for any vertex $v\in V(C_5)$, the set $\{C_5*v\}$ is also an obstruction set for the class 
of distance-hereditary graphs.

\begin{rem}\label{rem:rw-of-obstructions}\mbox{}\\[-7mm]
\begin{enumerate}
		\item Let $G$ be an obstruction for the class of all graphs of linear rankwidth at most $1$.
			Then $\rw(G)\leq\lrw(G) \leq 2$.
		\item The obstruction set for the class of all graphs of linear rankwidth at most~$1$
			is finite.
\end{enumerate}
\end{rem}
\begin{proof}
	For the first statement, let $G$ be such an obstruction, and let $v\in V(G)$. Then $\lrw(G\setminus v)\leq 1$
	by definition, and hence $\lrw(G)\leq 2$, because adding a vertex can increase the cutrank function
	by at most one.
	But $\lrw(G)\leq 2$ implies $\rw(G)\leq 2$. 

	The second statement follows from the first statement, together with the 
	fact that if there is a fixed upper bound on the rankwidth of all graphs in an obstruction set, 
	then the obstruction set is finite~\cite{Oum08}.
\qed\end{proof}

\medskip A vertex of degree $1$ in $G$ is called a \emph{pendant vertex} in $G$. We say that two distinct vertices
$u,v\in V(G)$ are \emph{strong siblings}, if 
$N_G(u)\setminus\{v\}=N_G(v)\setminus\{u\}$ 
and $\{u,v\}\in E(G)$.
Two distinct vertices
$u,v\in V(G)$ are \emph{weak siblings}, if $N_G(u)=N_G(v)$ 
(and $\{u,v\}\notin E(G)$).
A \emph{split pair} is a pair $u,v$ of vertices of $G$ that are either strong or
weak siblings. We will use the following fact.

\begin{fact}[\cite{BandeltM86}]\label{fact:dist-hereditary}
	Every finite distance-hereditary graph $G$ with at least four vertices has
	either at least two disjoint split pairs, or a split pair and a pendant vertex,
	or at least two pendant vertices.
\end{fact}


\section{Thread graphs}
Thread graphs were introduced in~\cite{Ganian10} as an alternative characterization
of graphs of linear rankwidth at most $1$. 
In this section we define thread graphs and we exhibit some of their properties.
We define thread graphs in a slightly
different way. It can be easily seen that our definition is indeed equivalent
to the original definition in~\cite{Ganian10}.

A \emph{thread block} is a tuple $(G,(a,b),\bar v,\LL)$, consisting of a graph $G$,
distinguished edge $\{a,b\}\in E(G)$, called the \emph{thread edge} of $G$, 
 an ordering $\bar v=v_1,\ldots,v_n$ of $V(G)$ with $v_1=a$ and $v_n=b$, called a \emph{thread ordering}, and 
 a \emph{thread labeling} $\LL\colon V(G)\to \big\{\{L\},\{R\},\{L,R\}\big\}$ of $V(G)$, such that
 $\text{for all }i,j\in\{1,\ldots,n\}$
\begin{itemize}
	\item $\LL(v_1)=\{R\}$, $\LL(v_n)=\{L\}$, and
	\item for all $1\leq i<j\leq n$,\\ 
		$\{v_i,v_j\}\in E(G)\text{ if and only if } R\in\LL(v_i)
	\text{ and }L\in\LL(v_j).$
\end{itemize}

Intuitively, every vertex $u$ with $L\in\LL(u)$
`sees' all vertices $v$ to its $L$eft that `look' to the right, i.e.\ that have $R\in\LL(v)$.
Symmetrically, every vertex $v$ with $R\in\LL(v)$
`sees' all vertices $u$ to its $R$ight that `look' to the left, i.e.\ that 
have $L\in\LL(u)$. 

Figure~\ref{fig:threadblock}  shows a graph $G$ with an edge $\{a,b\}$ and an ordering
$\bar v=a,a',v,u,b',b$ and a labeling $\LL$ such that $(G,(a,b), \bar v,\LL)$ is 
a thread block.

\begin{figure}
\begin{center}
\begin{tikzpicture}
	\vertexnodes 
\begin{scope}[scale=1]
	\coordinate(name) at (1.5,-1);
	\coordinate(aprime) at (0,1);
	\coordinate(a) at (1,1);
	\coordinate(b) at (2,1);
	\coordinate(bprime) at (3,1);
	\coordinate(u) at (1,0);
	\coordinate(v) at (2,0);
	\draw  (aprime) node[label=above:$a'$]{} -- (a) node[label=above:$a$]{} -- (b) node[label=above:$b$]{}
	-- (bprime) node[label=above:$b'$]{}; 
	\draw (u) node[label=below:$u$]{} --(v) node[label=below:$v$]{} -- (b); 
	\draw (a) -- (u) -- (b);
	\draw (name) node[fill=none]{};
\end{scope}
\end{tikzpicture}
\hspace{1.6cm}
\begin{tikzpicture}
	\vertexnodes 
\begin{scope}[scale=1.3, xshift=-1]
	\coordinate(aprime) at (0.5,0.5);
	\coordinate(a) at (0.5,1);
	\coordinate(b) at (2.5,1);
	\coordinate(bprime) at (2.5,0.5);
	\coordinate(u) at (2,0);
	\coordinate(v) at (1,0);
	\coordinate(inv) at (1,-0.62);
	\draw  (aprime) node[label=left:$a'$,label=below left:$\{L\}$]{} -- (a) node[label=above:$a$,label=left:$\{R\}$]{} -- (b) node[label=above:$b$,label=right:$\{L\}$]{}
	-- (bprime) node[label=right:$b'$,label=below right:$\{R\}$]{}; 
	\draw (u) node[label=below:$u$,label={[name=label node]below:$\{L,R\}$}]{} --
	(v) node[label=below:$v$]{} 
	-- (b); 
	\draw (a) -- (u) -- (b);
	\node[fill=none] at (inv){$\{R\}$};
\end{scope}
\end{tikzpicture}
\end{center}
\caption{A graph $G$ and a thread block $(G,(a,b),a,a',v,u,b',b,\LL)$.}\label{fig:threadblock}
\end{figure}
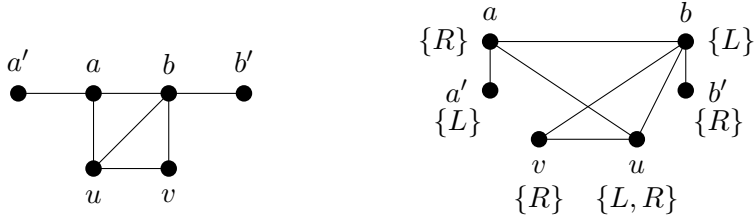

\begin{rem}
	Let $(G,(a,b),\bar v,\LL)$ be a thread block. 
	Then the set of vertices $\LL^{-1}(\{L,R\})\subseteq V(G)$
	induces a complete graph in $G$.
\end{rem}

Let $(G,(a,b),\bar v,\LL)$ be a thread block, with $\bar v=v_1,\ldots,v_n.$ 
For $1\leq i\leq j\leq n$ let $v_i,\ldots ,v_j$ be an interval of $\bar v$.
The interval $v_i,\ldots ,v_j$ is called $\LL$\emph{-constant}, if
every $\ell \in\{i,\ldots,j\}$ satisfies $\LL(v_{\ell})=\LL(v_{i})$.

\begin{lem}[Permuting thread orderings]\label{lem:reorder-vertices}
	Let $(G,(a,b),\bar v,\LL)$ be a thread block.  	
	For $1\leq i\leq j\leq n$ let $v_i,\ldots ,v_j$ be an
	$\LL$-constant
	interval of $\bar v$. For a permutation $\pi(v_i,\ldots v_j)$ of $v_i,\ldots v_j$, 
	let $\bar u$ be the ordering of $V(G)$ obtained by replacing the interval $v_i,\ldots ,v_j$
	in $\bar v$ by the interval $\pi(v_i,\ldots v_j)$.
	Then $(G,(u_1,u_n),\bar u,\LL)$ is a thread block as well.
\end{lem}
\begin{proof}
	Observe that any two vertices in an $\LL$-\emph{constant} interval form a split pair,
	and hence transposing them yields a thread ordering. Since any permutation is a
	product of transpositions, this proves the lemma.
\qed\end{proof}

\medskip 
A \emph{connected thread graph} is a graph $G$ that either consists of a single vertex only, or is obtained from a 
sequence \[(G_1,(a_1,b_1),\bar v^1,\LL^1),\ldots ,(G_m,(a_m,b_m),\bar v^m,\LL^m)\]
of thread blocks $(G_i,(a_i,b_i),\bar v^i,\LL^i)$,
for $i\in[m]$, by identifying
$b_i$ with $a_{i+1}$ for all $i\in[m-1]$.
The path $a_1,\ldots,b_m$ in $G$ of length $m$ 
thus obtained is called the \emph{thread} of $G$.
A \emph{thread graph} is either the empty graph, 
or a disjoint union of connected thread graphs.

The following theorem was proven in~\cite{Ganian10}. We give a brief proof here for completeness.

\begin{theo}[Ganian~\cite{Ganian10}]\label{theo:ganian}
	A graph $G$ has $\lrw(G)\leq 1$ if and only if $G$ is a thread graph.
\end{theo}
\begin{proof}
	We may assume that $G$ is connected and $E(G)\neq\emptyset$.  
	We define a \emph{thread ordering 
	of a connected thread graph} $G$
	to be the concatenation of thread orderings of
	a sequence of thread blocks that yield $G$, 
	identifying the shared thread vertices and labeling
	them $\{L,R\}$. Other thread labels are as determined for the thread blocks.

	Assume  $\lrw(G)\leq 1$, and let
	$(T,\lambda)$ be a linear rank decomposition witnessing this.   Consider a
	total ordering $\prec$ of the vertices of $G$ that is consistent with  the
	linear structure of $T$ yielding the linear rankwidth $\leq 1$.  
	We will prove that this ordering is a thread ordering.
	Consider further a vertex $v$ being {\it  processed}. 
	There is a unique binary string expressing adjacencies
	between already processed vertices $u\prec v$ and the vertices $w, v \preceq w$.  

  We use ${\bf e}=0^*$ to represent the pattern of all
	$0$'s, i.e., no adjacencies (``empty neighborhood''). We use
	${\bf n=e}1\{0,1\}^*$ to mean an arbitrary pattern of $0$'s and $1$'s,  including at least
	one $1$ and perhaps no $0$'s. 
	
	\emph{Case 1}: The neighborhood of processed vertices
		is $1{\bf n}$.
		Since $v$ is the first unprocessed vertex it is  adjacent to the
		processed vertices.  After $v$ is processed, it could  either have no
		adjacencies to the remaining unprocessed vertices, in  which case we label it
		$\{L\}$ in the corresponding thread ordering, or the neighborhood
		could  be the same as the neighborhood of other processed vertices, in
		which case it is labeled $\{L,R\}$ in the corresponding thread ordering.

		\emph{Case 2}: The neighborhood of processed vertices is
		$1{\bf e}$. This identifies $v$ as a thread vertex. After processing,  $v$ has either
		an empty neighborhood, in which case we label it $\{L\}$,
		or its adjacencies with unprocessed vertices are expressed by ${\bf n}$, 
		in which case $v$ is labeled $\{L,R\}$ and is an internal
		thread vertex.

		\emph{Case 3}: The neighborhood is $0{\bf n}$.  After $v$ is  processed, it must have a neighborhood
		${\bf n}$ as do other
		processed vertices, in which case it is labeled $\{R\}$ in the thread
		ordering. This is a thread labeling proving that $G$ is a thread graph.

		For the converse, assume that $G$ is a thread graph with a 
		given thread ordering $\prec$ of $V(G)$. We define a linear
		rank decomposition $(T,\lambda)$ by mapping the leaves of $T$ to the vertices of
		$G$ in such a way that the linear structure of $(T,\lambda)$ respects $\prec$.
		It is straightforward to verify that the width of $(T,\lambda)$ is $\leq 1$.
		\qed\end{proof}

		\begin{rem}[Basic properties of thread graphs]\label{rem:thread-basics}\mbox{}\\[-7mm] 
	\begin{enumerate}
		\item Let $G$ be a connected thread graph with thread $a_1,\ldots,a_m$ and 
			let $X\subseteq V(G)$ be the set of all cut-vertices of $G$. Then
			$X\subseteq\{a_1,\ldots a_m\}$ and $\{a_2,\ldots a_{m-1}\}\subseteq X$.
		\item Let $G$ be a $2$-connected thread graph. Then every thread in $G$
			consists of a single edge. 
		\item 
			Let $G$ be a connected thread graph obtained from the sequence
			\[S:=(G_1,(a_1,b_1),\bar v^1,\LL^1),\ldots ,(G_m,(a_m,b_m),\bar
			v^m,\LL^m)\] of thread blocks, where $P=a_1,\ldots,b_m$ is a thread in $G$.
			Then, for any interval $S'$ of $S$, the thread graph $G'$ obtained
			from $S'$ is a connected induced subgraph of $G$ with thread $P':=P\cap G'$.
		\item Every thread in a connected thread graph
			$G$ is an induced path in $G$.
\end{enumerate}
\end{rem}
\begin{proof}
	The first statement is proved in~\cite{Ganian10}, and it implies the second statement.
	The last two statements follow from the definition of connected thread graphs.
\qed\end{proof}

\begin{lem}[Removing pendant vertices]\label{lem:safe-red}
	Let $G$ be a graph, let 
			$u\in V(G)$ be a pendant vertex with unique neighbor $c\in V(G)$, such that
			$c$ is a cut-vertex of $G\setminus u$. 
			Then $G$ is a thread graph if and only if $G\setminus v$ is a thread graph.
\end{lem}
\begin{proof}	
	If $G$ is a thread graph, then, using the equivalence between thread graphs and
	graphs of linear rankwidth at most $1$ (Theorem~\ref{theo:ganian}), by
	Lemma~\ref{lem:vm-closed}.3 the graph $G\setminus v$ is a thread graph.
	
	Conversely, let $G\setminus v$ be a thread graph. Since $c$ is a cut-vertex in 
	$G\setminus v$, $c$ lies on every thread. 
	Choose a thread block of $G$ containing $c$. Then $c$ is either the
	first or the last vertex in the thread ordering of that thread block.
	If $c$ is the first vertex, add $u$ immediately after
	$c$ to the thread ordering and label it $\{L\}$. Symmetrically,
	if $c$ is the last vertex, add $u$ immediately before
	$c$ to the thread ordering and label it $\{R\}$. Hence $G$ is a thread graph.
\qed\end{proof}

\begin{lem}[Thread graphs with two `whiskers']\label{lem:pending-vertices-help}
		Let $G$ be a thread graph. Assume that $G$ contains a $2$-connected 
		subgraph $G_0\subseteq G$ and two vertices
		$u,v\in V(G)$ such that $V(G)= V(G_0)\cup\{u,v\}$, and $u$ and $v$ are pendant vertices in $G$.
		Let $a\in V(G_0)$ be the neighbor of $u$ and let $b\in V(G_0)$ be the neighbor of $v$,
		and assume that $a\neq b$. Then $\{a,b\}\in E(G)$ and there exists a
		thread ordering $\bar v$ of $V(G)$ and a labeling $\LL$ such that $(G,(a,b),\bar v,\LL)$ 
		is a thread block.
	\end{lem}
	\begin{proof}
		Let $P$ be a thread in $G$. Since $a$ and $b$ are cut vertices, by Remark~\ref{rem:thread-basics}.1,
		$a$ and $b$ lie on $P$. By Remark~\ref{rem:thread-basics}.3, $P\cap G_0$ is a thread for $G_0$,
		and by Remark~\ref{rem:thread-basics}.2, $P\cap G_0$
		consists of a single edge $e$ only. Since $a,b\in V(G_0)\cap V(P)$,
		$e=\{a,b\}$ and the lemma follows.
	\qed\end{proof}

\section{The obstruction set for linear rankwidth at most~$1$}

From now on, let $\CC:=\{G\text{ graph}\mid \lrw(G)\leq 1\}$ 
denote the class of all graphs of linear rankwidth at most $1$.
We first show that the graphs $C_5$, $N$ and $Q$ shown in Figure~\ref{fig:obstructions}
are obstructions for $\CC$.
The harder part will be to show that the set $\{C_5, N, Q\}$ is the complete obstruction set.
\begin{lem}\label{lem:they-are-obs}
	The three graphs $C_5$, $N$ and $Q$ 
	are obstructions for the class of all graphs of linear rankwidth at most~$1$.
\end{lem}

\noindent\emph{Proof sketch.} 
	We first have to show that none of the three graphs $C_5$, $N$ and $Q$
	have linear rankwidth $1$ (which, by Lemma~\ref{lem:vm-closed}.2,
	implies that no graph locally equivalent to $C_5$, $N$ or $Q$ has linear rankwidth $1$). 
	Second, for every graph $H$ that is locally equivalent to 
	one of the three graphs $C_5$, $N$ and $Q$, and for every $v\in V(H)$, we have to show that
	$\lrw(H\setminus v)\leq 1$. 
	
	For the first part, we have already seen in Example~\ref{ex:seefive} that
	$\lrw(C_5)=2$. 
	Using Theorem~\ref{theo:ganian}, it suffices to show that neither $N$ nor $Q$ contains a thread.
	By Remark~\ref{rem:thread-basics}, every thread in $N$ would contain the three cut-vertices, 
	but the cut-vertices do not lie on an induced path, which is necessary by 
	Remark~\ref{rem:thread-basics}.4. Similarly, the two cut-vertices of $Q$ would
	have to lie on every thread, but there is no path connecting them that only
	uses cut-vertices, which would be necessary by Remark~\ref{rem:thread-basics}.1.
	Hence $C_5$, $N$ and $Q$ are not thread graphs.
	
	For the second part, using Theorem~\ref{theo:ganian}, for every graph $H$ that is 
	locally equivalent to 
	one of the three graphs $C_5$, $N$ and $Q$, and for every $v\in V(H)$,
	one has to exhibit a 
	thread, a thread ordering and a corresponding labeling.
	This is not hard to do and is left to the reader. 
	Figures~\ref{fig:C5equiv}, \ref{fig:N-equivalent}, and~\ref{fig:Q-equivalent} show the
	classes of graphs that are locally equivalent to $C_5$, $N$, and $Q$, respectively.
\qed

\begin{figure}[h]
\begin{center}	
\begin{tikzpicture}
\begin{scope}[scale=.7]
	\vertexnodes 
	\coordinate(top) at (90:0.7cm);
	\coordinate(left) at (210:0.7cm);
	\coordinate(right) at (-30:0.7cm);
	\coordinate(extop) at (90:1.5cm);
	\coordinate(exleft) at (210:1.5cm);
	\coordinate(exright) at (-30:1.5cm);
	\draw  (top) node{} -- (left) node{} -- (right) node{} -- cycle;
	\draw (top) -- (extop) node{};
	\draw (left) -- (exleft) node{};
	\draw (right) -- (exright) node{};
\end{scope}
\end{tikzpicture}
\hspace{1cm}
\begin{tikzpicture}
	\vertexnodes 
\begin{scope}[scale=.7]
	\coordinate(aprime) at (0,1);
	\coordinate(a) at (1,1);
	\coordinate(x) at (2,2);
	\coordinate(y) at (2,0);
	\coordinate(b) at (3,1);
	\coordinate(bprime) at (4,1);
	\draw  (aprime) node{} -- (a) node{} -- (x) node{} --(b) node{} --(bprime) node{}; 
	\draw (y) node{} -- (a);
	\draw (y)  -- (x);
	\draw (y)  -- (b);
\end{scope}
\end{tikzpicture}
\hspace{1cm}
\begin{tikzpicture}
	\vertexnodes 
\begin{scope}[scale=.7]
	\coordinate(aprime) at (0,1);
	\coordinate(a) at (1,1);
	\coordinate(x) at (2,2);
	\coordinate(y) at (2,0);
	\coordinate(b) at (3,1);
	\coordinate(bprime) at (4,1);
	\draw  (aprime) node{} -- (a) node{} -- (x) node{} --(b) node{} --(bprime) node{}; 
	\draw (y) node{} -- (a);
	\draw (y)  -- (b);
	\draw (aprime)  -- (x);
	\draw (aprime)  -- (y);
\end{scope}
\end{tikzpicture}
\hspace{1cm}
\begin{tikzpicture}
	\vertexnodes 
\begin{scope}[scale=.7]
	\coordinate(aprime) at (0,1);
	\coordinate(a) at (1,1);
	\coordinate(x) at (2,2);
	\coordinate(y) at (2,0);
	\coordinate(b) at (3,1);
	\coordinate(bprime) at (4,1);
	\draw  (aprime) node{} -- (a) node{} -- (x) node{} --(b) node{} --(bprime) node{}; 
	\draw (y) node{} -- (a);
	\draw (y)  -- (b);
	\draw (y)  -- (x);
	\draw (aprime)  -- (x);
	\draw (aprime)  -- (y);
	\draw (bprime)  -- (x);
	\draw (bprime)  -- (y);
\end{scope}
\end{tikzpicture}
\hspace{1cm}
\begin{tikzpicture}
	\begin{scope}[scale=1]
	\vertexnodes 
	\coordinate(a) at (0,0);
	\coordinate(aprime) at (-0.3,0.3);
	\path (0,0) coordinate (origin); 
	\path (0:1cm) coordinate (P0);
	\path (1*90:1cm) coordinate (P1);
	\path (2*90:1cm) coordinate (P2);
	\path (3*90:1cm) coordinate (P3);

	\draw (P0) node{} -- (P1) node{} -- (P2) node{} -- (P3) node{} -- cycle; 
	\draw (aprime) node{} -- (a) node{};
	\draw (a) -- (P0);
	\draw (a) -- (P1);
	\draw (a) -- (P2);
	\draw (a) -- (P3);
	\end{scope}
\end{tikzpicture}
\end{center}
\caption{The five graphs that are locally equivalent to the net graph $N$.}\label{fig:N-equivalent}
\end{figure}
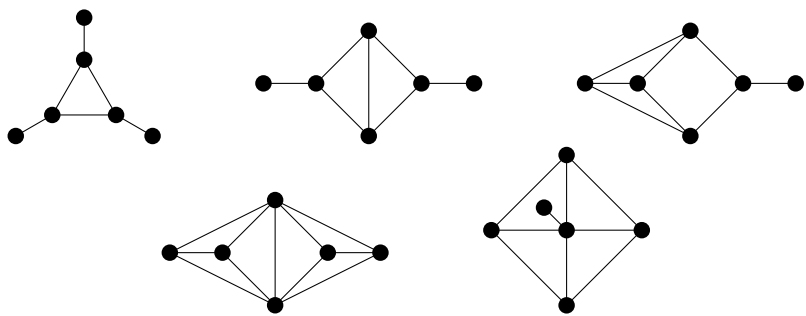

\begin{figure}[h]
\begin{center}
\begin{tikzpicture}
	\vertexnodes 
\begin{scope}[scale=.7]
	\coordinate(aprime) at (0,1);
	\coordinate(a) at (1,1);
	\coordinate(x) at (2,2);
	\coordinate(y) at (2,0);
	\coordinate(b) at (3,1);
	\coordinate(bprime) at (4,1);
	\draw  (aprime) node{} -- (a) node{} -- (x) node{} --(b) node{} --(bprime) node{}; 
	\draw (y) node{} -- (a);
	\draw (y)  -- (b);
\end{scope}
\end{tikzpicture}
\hspace{1cm}
\begin{tikzpicture}
	\vertexnodes 
\begin{scope}[scale=.7]
	\coordinate(aprime) at (0,1);
	\coordinate(a) at (1,1);
	\coordinate(x) at (2,2);
	\coordinate(y) at (2,0);
	\coordinate(b) at (3,1);
	\coordinate(bprime) at (4,1);
	\draw  (aprime) node{} -- (a) node{} -- (x) node{} --(b) node{} --(bprime) node{}; 
	\draw (y) node{} -- (a);
	\draw (a)  -- (b);
	\draw (y)  -- (b);
\end{scope}
\end{tikzpicture}
\hspace{1cm}
\begin{tikzpicture}
	\vertexnodes 
\begin{scope}[scale=.7]
	\coordinate(aprime) at (0,1); 
	\coordinate(a) at (1,1); 
	\coordinate(x) at (2,2); 
	\coordinate(y) at (2,0);
	\coordinate(b) at (3,1); 
	\coordinate(bprime) at (4,1); 
	\draw  (aprime) node{} -- (a) node{} -- (x) node{} --(b) node{} --(bprime) node{}; 
	\draw (y) node{} -- (a);
	\draw (y)  -- (b);
	\draw (x)  -- (y);
	\draw (aprime)  -- (x);
	\draw (aprime)  -- (y);
\end{scope}
\end{tikzpicture}
\hspace{1cm}
\begin{tikzpicture}
	\vertexnodes 
\begin{scope}[scale=.7]
	\coordinate(aprime) at (0,1);
	\coordinate(a) at (1,1);
	\coordinate(x) at (2,2);
	\coordinate(y) at (2,0);
	\coordinate(b) at (3,1);
	\coordinate(bprime) at (4,1);
	\draw  (aprime) node{} -- (a) node{} -- (x) node{} --(b) node{} --(bprime) node{}; 
	\draw (y) node{} -- (a);
	\draw (y)  -- (b);
	\draw (aprime)  -- (x);
	\draw (aprime)  -- (y);
	\draw (bprime)  -- (x);
	\draw (bprime)  -- (y);
\end{scope}
\end{tikzpicture}
\hspace{1cm}
\begin{tikzpicture}
	\begin{scope}[scale=1]
	\vertexnodes 
	\coordinate(a) at (0,0);
	\coordinate(aprime) at (-0.3,0.3);
	\coordinate(b) at (0.3,-0.3);
	\path (0,0) coordinate (origin); 
	\path (0:1cm) coordinate (P0);
	\path (1*90:1cm) coordinate (P1);
	\path (2*90:1cm) coordinate (P2);
	\path (3*90:1cm) coordinate (P3);

	\draw (P0) node{} -- (P1) node{} -- (P2) node{} -- (P3) node{} -- cycle; 
	\draw (b) node{} -- (P0);
	\draw (b) -- (P1);
	\draw (b) -- (P2);
	\draw (b) -- (P3);
	\draw (aprime) node{} -- (P0);
	\draw (aprime) -- (P1);
	\draw (aprime) -- (P2);
	\draw (aprime) -- (P3);
	\end{scope}
\end{tikzpicture}
\end{center}
\caption{The five graphs that are locally equivalent to the half-cube $Q$.}\label{fig:Q-equivalent}
\end{figure}
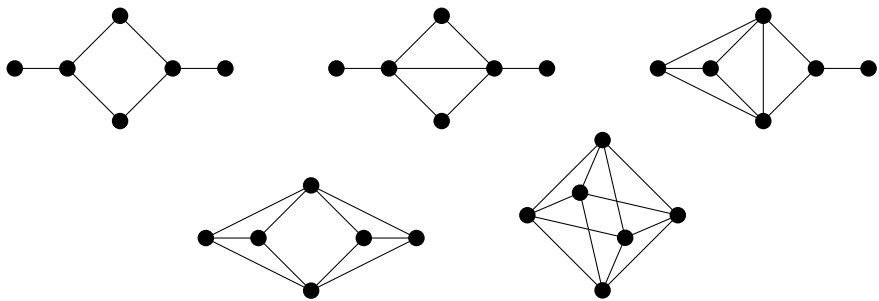

\begin{lem}\label{lem:obstr-are-connected}
	Every obstruction $G$ for $\CC$ is connected.
\end{lem}
\begin{proof}
	By contradiction.  Suppose that $G=G_1\dotcup G_2$, where
	$G_1$ and $G_2$ are non-empty unions of connected components of $G$.
	Let $x_1\in V(G_1)$ and $x_2\in V(G_2)$. Since $G$ is an obstruction, $G\setminus x_1$ 
	is a thread graph, implying $G_2$ is a thread graph. Symmetrically, 
	$G\setminus x_2$  is a thread graph, implying  $G_1$ is a thread graph. But the disjoint union of
	thread graphs is a thread graph, so $G$ is a thread graph as well, a contradiction.
\qed\end{proof}

\begin{lem}\label{lem:cut-into-three}
	If $G$ is an obstruction for $\CC$, then $G$ does not contain
	a cut-vertex 
	$v\in V(G)$ such that $G\setminus v$ has more than two components.
\end{lem}
\begin{proof}
	By contradiction. Assume that $G$ is an obstruction
	containing a cut-vertex $v$ such that $G\setminus v$ has three components
	$C_1,C_2$ and $C_3$.
	By Lemma~\ref{lem:obstr-are-connected}, $G$ is connected, so every component
	$C_i$ contains a neighbor $x_i$ of $v$ in $G$, for $i\in[3]$.
	
	Every component $C_i$ contains at least two vertices.
	Otherwise, a component $C_i$ with one vertex $x_i$ would be a pendant vertex
	attached to $v$, and $v$ is a cut-vertex in $G\setminus x_i$, 
	and since $G$ is an obstruction,
	$G\setminus x_i$ cannot be a thread graph by Lemma~\ref{lem:safe-red}.
   
	Consequently, every component $C_i$ of $G\setminus v$ 
	contains a neighbor $y_i$ of  $x_i$.  
	The vertices $v,x_i,y_i$ either induce a path of length two or a triangle $K_3$
	in $G$.
	If they induce a $K_3$ in $G$, then they induce a path of length two in $G*x_i$
	(see Figure~\ref{fig:cut-into-three}).
	Hence we may assume that for all $i\in[3]$, the vertices $v,x_i,y_i$ 
	induce a path of length two in $G$. 
	But then $\big(G[\{v,x_1,x_2,x_3,y_1,y_2,y_3\}]*v\big)\setminus v$ 
	 s isomorphic to the net graph $N$ (see Figure~\ref{fig:cut-into-three-2}), 
	a contradiction to $G$ being an obstruction.
\qed\end{proof}


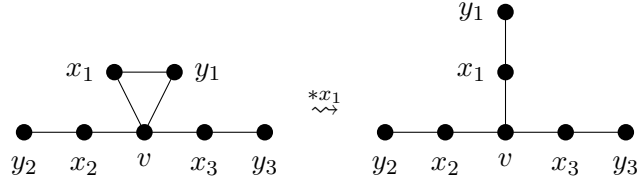
\begin{figure}[h]
\begin{center}	
\begin{tikzpicture}
	\vertexnodes 
\begin{scope}[scale=.8]
	\coordinate(v) at (2,0);
	\coordinate(x1) at (1.5,1);
	\coordinate(x2) at (1,0);
	\coordinate(x3) at (3,0);
	\coordinate(y1) at (2.5,1);
	\coordinate(y2) at (0,0);
	\coordinate(y3) at (4,0);
	\coordinate(arrow) at (5,0.5);
	\draw (x1) node[label=left:$x_1$]{} -- (y1) node[label=right:$y_1$]{} -- (v) node[label=below:$v$]{} --cycle;
	\draw (y2)node[label=below:$y_2$]{}  -- (x2) node[label=below:$x_2$]{} -- (v) --
	(x3)node[label=below:$x_3$]{}  -- (y3) node[label=below:$y_3$]{}; 
	\draw (arrow) node[fill=none]{$\stackrel{*x_1}{\leadsto}$};
\end{scope}
\begin{scope}[scale=.8, xshift=6cm]
	\coordinate(v) at (2,0);
	\coordinate(x1) at (2,1);
	\coordinate(x2) at (1,0);
	\coordinate(x3) at (3,0);
	\coordinate(y1) at (2,2);
	\coordinate(y2) at (0,0);
	\coordinate(y3) at (4,0);
	\draw (y1) node[label=left:$y_1$]{} -- (x1) node[label=left:$x_1$]{} -- (v) node[label=below:$v$]{};
	\draw (y2)node[label=below:$y_2$]{}  -- (x2) node[label=below:$x_2$]{} -- (v) --
	(x3)node[label=below:$x_3$]{}  -- (y3) node[label=below:$y_3$]{}; 
\end{scope}
\end{tikzpicture}
\end{center}
\caption{Producing an induced path of length $2$ in the proof of Lemma~\ref{lem:cut-into-three}.}
\label{fig:cut-into-three}
\end{figure}

\begin{figure}[h]
\begin{center}	
\begin{tikzpicture}
	\vertexnodes 
\begin{scope}[scale=1]
	\coordinate(v) at (1.5,0.5);
	\coordinate(x1) at (1.5,1);
	\coordinate(x2) at (0.8,0);
	\coordinate(x3) at (2.2,0);
	\coordinate(y1) at (1.5,1.8);
	\coordinate(y2) at (0,0);
	\coordinate(y3) at (3,0);
	\draw (x1) node[label=left:$x_1$]{} -- (x2) node[label=below:$x_2$]{} -- (x3) node[label=below:$x_3$]{} --cycle;
	\draw (y2) node[label=below:$y_2$]{}  -- (x2) -- (x3) -- (y3) node[label=below:$y_3$]{};
	\draw (y1) node[label=left:$y_1$]{}  -- (x1);
	\draw (v) node[label=below:$v$]{}  -- (x1);
	\draw (v)  -- (x2);
	\draw (v)  -- (x3);
\end{scope}
\end{tikzpicture}
\end{center}
\caption{The graph $G[\{v,x_1,x_2,x_3,y_1,y_2,y_3\}]*v$ in the proof of Lemma~\ref{lem:cut-into-three}.}
\label{fig:cut-into-three-2}
\end{figure}
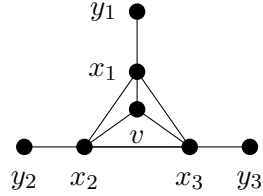

\medskip
We will investigate the structure of obstructions by considering their \emph{block-cut-vertex trees}.
We say that a
graph is \emph{non-separable}, if it is non-trivial, 
connected and contains no cut-vertices. The non-separable 
graphs are
$K_2$ and all $2$-connected graphs.
A \emph{block} in $G$ is a subgraph of $G$ that is non-separable and 
maximal with respect to this property.
Given a connected graph $G$, we define a bipartite graph $\BB(G)$ with vertex set $V(\BB(G))=X\dotcup Y$,
where $X$ is the set of all cut-vertices of $G$ and $Y$ is the
set of all blocks in $G$ and there is an edge from $x\in X$ to $y\in Y$
if and only if $x\in y$. It is well-known that 
$\BB(G)$ is a tree~\cite{harary69}. The tree $\BB(G)$ is called the 
\emph{block-cut-vertex tree} of $G$.
By Lemma~\ref{lem:obstr-are-connected} every obstruction $G$ for $\CC$ is connected, 
so $G$ has a block-cut-vertex tree.

\begin{lem}\label{lem:blockcut-into-three}
	Let $G$ be an obstruction and let $B$ be a block in $G$.
	If $\deg_{\BB(G)}(B)\geq 3$, then $G$ is isomorphic to $N$.
\end{lem}
\begin{proof}
	Let $\deg_{\BB(G)}(B)\geq 3$, and let $a_1,a_2,a_3$ be three neighbors of
	$B$ in the block-cut-vertex graph $\BB(G)$. Then $a_1,a_2,a_3$ are cut-vertices
	in $G$. 

	The vertices $a_1,a_2,a_3$ are pairwise distinct:
	Otherwise, if, $a_i=a_j$ for some $i,j\in[3]$ with $i\neq j$, 
	then $G\setminus a_i$ has at least three connected components, and by Lemma~\ref{lem:cut-into-three},
	$G$ is not an obstruction.
	Since $\{a_1,a_2,a_3\}\subseteq V(B)$ we have $\left|V(B)\right|\geq 3$ and 
	hence $B$ is $2$-connected. Since $a_i$ is a cut-vertex, for every $i\in[3]$
	there exists a neighbor $b_i$ of $a_i$, $b_i\in V(G)\setminus V(B)$.
	The graph $G\setminus b_i$ is a thread graph, and by 
	Lemma~\ref{lem:pending-vertices-help}, the edge $\{a_j,a_k\}$ is a thread in
	$G[V(B)\cup\{b_j,b_k\}]$ for $\{j,k\}=[3]\setminus\{i\}$. In particular,
	$\{a_j,a_k\}\in E(G)$ for all pairs $j,k\in[3]$ with $j\neq k$. 
	Hence $V(G)=\{a_1,a_2,a_3,b_1,b_2,b_3\}$ and $G$ is isomorphic to $N$.
\qed\end{proof}

\begin{lem}\label{lem:cut-into-two}
	Let $G$ be an obstruction and let $v$ be a cut-vertex in $G$.
	Then one of the connected components of $G\setminus v$ is trivial.
\end{lem}
\begin{proof}
	By Lemma~\ref{lem:cut-into-three}, $G\setminus v$ has exactly two connected
	components.
	Towards a contradiction, assume
	that both components $C_1$ and $C_2$ of $G\setminus v$
	contain at least two vertices. Since $G$ is connected, there exist
	vertices $x_i\in V(C_i)$ that are neighbors of $v$ in $G$, and
	let $y_i\in V(C_i)\setminus\{x_i\}$ (for $i\in[2]$). Then $G\setminus y_i$
	is a thread graph with a thread $P_i$. Since $v$ is a cut-vertex in $G\setminus y_i$,
	$v$ lies on $P_i$ by Remark~\ref{rem:thread-basics}.1. By 
	Remark~\ref{rem:thread-basics}.3, $P_i':=P_i\cap G[V(C_i)\cup\{v\}]$
	is a thread in $G[V(C_i)\cup\{v\}]$ and $v\in V(P_i')$. But since
	$C_i$ is connected, $v$ is not a cut-vertex in $G[V(C_i)\cup\{v\}]$,
	hence $v$ is an end-vertex of $P_i'$ by Remark~\ref{rem:thread-basics}.1.
	But $V(P_1')\cap V(P_2')=\{v\}$, and we find that $V(P_1')\cap V(P_2')$
	is a thread for $G$, which contradicts $G$ being an obstruction.
	\qed\end{proof}

\begin{theo}[Classification]\label{theo:classify-obstrucitons}
	Let $G$ be an obstruction for $\CC$. Then either 
	\begin{enumerate}
		\item $G$ is isomorphic to $N$, or
		\item $G$ is $2$-connected.
		\item $G$ has a $2$-connected subgraph $G_0\subseteq G$  and a vertex
			$u\in V(G)$ such that $V(G)= V(G_0)\cup\{u\}$ and $u$ is a pendant vertices in $G$, or	
		\item $G$ contains a $2$-connected subgraph $G_0\subseteq G$ and two vertices
			$u,v\in V(G)$ such that $V(G)= V(G_0)\cup\{u,v\}$ and $u$ and $v$ are pendant vertices in $G$
			with $N_G(u)\cap N_G(v)=\emptyset$, or
	\end{enumerate}
\end{theo}

\begin{proof} Lemma~\ref{lem:they-are-obs} proves Case~1.
	Let $G$ be an obstruction for $\CC$ and assume that $G$ is not isomorphic to $N$. 
	Then, by Lemmas~\ref{lem:cut-into-three} and~\ref{lem:blockcut-into-three},  
	the block-cut-vertex tree $\BB(G)$ is a path. 
	By Remark~\ref{rem:four-verts}, $G$ has more than four vertices.
	By Lemma~\ref{lem:cut-into-two}, the path $\BB(G)$ has length $0$, $2$ or $4$. 
	If $\BB(G)$ has length $0$, then $G$ is a block and hence $G$ is two-connected,
	which is Case~2.

	If $\BB(G)$ has length $2$, then $\BB(G)$ has two blocks, and by Lemma~\ref{lem:cut-into-two}, one block
	of $\BB(G)$ is isomorphic to $K_2$. Then the other block has at least four vertices, and hence is
	$2$-connected. This is Case~3.
	
	If $\BB(G)$ has length $4$, then $\BB(G)$ has three blocks. 
	By Lemma~\ref{lem:cut-into-two}, the two blocks $y_1$ and $y_2$ of degree one in $\BB(G)$ are
	both isomorphic to $K_2$. Hence the third block $y_3$ contains at least three vertices,
	and therefore it is $2$-connected. By Lemma~\ref{lem:cut-into-three}, the intersection
	$V(y_1)\cap V(y_2)=\emptyset$ is empty.
	This is Case~4.
\qed\end{proof}

\begin{lem}[Two `whiskers' in obstructions]\label{lem:wiskers}
	Let $H$ be an obstruction for $\CC$ such that $H\not\sim C_5$ and $H\not\sim N$.
	Then $H$ is locally equivalent to an obstruction $G$ 
	that contains a $2$-connected subgraph $G_0\subseteq G$ and two vertices
	$u,v\in V(G)$ such that $V(G)= V(G_0)\cup\{u,v\}$ and $u$ and $v$ are 
	pendant vertices in $G$
	with $N_G(u)\cap N_G(v)=\emptyset$ (i.e.\ Case~4 in 
	Theorem~\ref{theo:classify-obstrucitons}).
\end{lem}
\begin{proof}
	Let $H$ be an obstruction for $\CC$ such that $H\not\sim C_5$ and $H\not\sim N$.
	Then $C_5$ is not a vertex-minor of $H$, and $H$ is distance-hereditary by Fact~\ref{fact:seefive}.
	Since every graph on four vertices is a thread graph, we know that $\left|V(H)\right|>4$.
  	Therefore we can apply Fact~\ref{fact:dist-hereditary} to find that 
	$G$ has either at least two disjoint split pairs, or a split pair and a pendant vertex,
	or at least two pendant vertices.

	We show that any split pair $u,v\in V(H)$ can be transformed into a pendant vertex by applying
	local complementations.
	If $u$ and $v$ are strong siblings, then a local complementation at $u$ 
	results in $v$ being a pendant vertex.
	If $u$ and $v$ are weak siblings, then, since $H$ is connected, there is a 
	vertex $w\in N_H(u)\cap N_H(v)$. It is easy to see that a local complementation 
	at $w$ transforms $u$ and $v$ into strong siblings in the locally equivalent graph.

	Hence $H$ is locally equivalent to a graph $G$ with two pendant vertices.
	Since $H$ is an obstruction, by definition, $G$ is an obstruction as well,
	so $G$ must be of the form of Case~4 in Theorem~\ref{theo:classify-obstrucitons}.
\qed\end{proof}

\medskip
We are now ready to prove the main result of this paper.

\medskip\emph{Proof of Theorem~\ref{theo:main}.}
	By Lemma~\ref{lem:they-are-obs}, $C_5$, $N$ and $Q$ are obstructions
	for $\CC$.
	We show that if $G$ is an obstruction for $\CC$ such that 
	$G\not\sim C_5$ and $G\not\sim N$, then $G\sim Q$. 
	By Lemma~\ref{lem:wiskers} we may assume that
	$G$ contains a $2$-connected subgraph $G_0\subseteq G$ and two vertices
	$a',b'\in V(G)$ such that $V(G)= V(G_0)\cup\{a',b'\}$ and $a'$ and $b'$ are pendant vertices in $G$
	with $N_G(a')\cap N_G(b')=\emptyset$. Let $a\in V(G)$ 
	be the unique neighbor of $a'$ and let $b\in V(G)$ 
	be the unique neighbor of $b'$. 

	Since $G$ is an obstruction, the graph $G\setminus b'$ is a thread graph. 
	Since $a$ is a cut-vertex in $G\setminus b'$,
	vertex $a$ lies on every thread of $G\setminus b'$ (by Remark~\ref{rem:thread-basics}.1). 
	Using Remark~\ref{rem:thread-basics} and the fact that $G_0$ is $2$-connected,
	it is easy to see that there is an edge $\{a,c\}\in E(G\setminus b')$,
	an ordering $\bar v=v_1,\ldots,v_n$ of $V(G\setminus b')$ with $a=v_1$ and 
	$c=v_n$, and a labeling $\LL$ such that 
	$(G\setminus b',(a,c),\bar v,\LL)$ is a thread block. Moreover, $b\neq c$, 
	because otherwise $G$ would be a thread graph. (To see this, place $b'$ between $v_{n-1}$ and $v_n$
	and let $\LL(b'):=\{R\}$.)
	Let us consider different cases of the label $\LL(b)$ in $\bar v$.
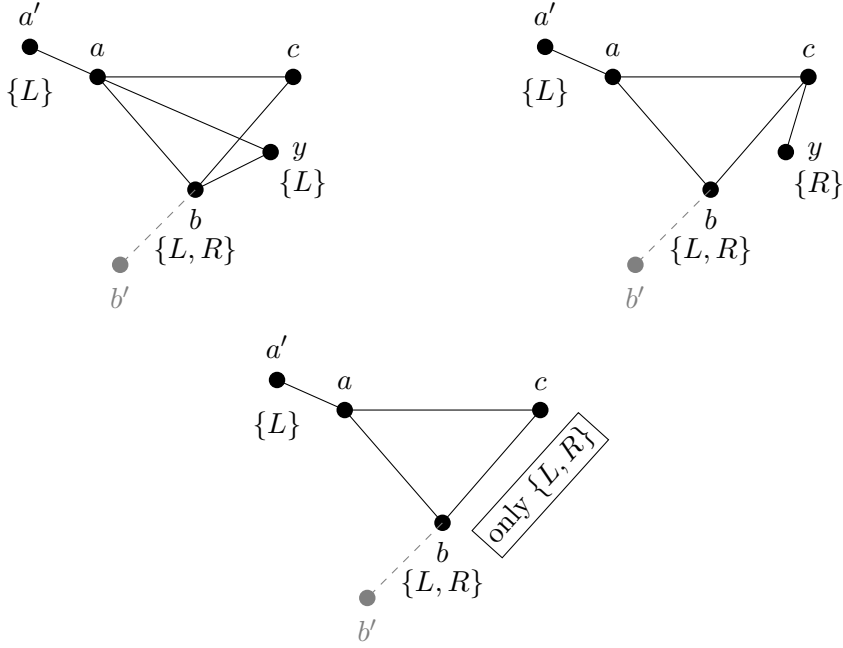
\begin{figure}[h]
\begin{center}	
\begin{tikzpicture}
	\vertexnodes 
\begin{scope}[scale=1]
	\coordinate(aprime) at (0.8,2.9);
	\coordinate(a) at (1.7,2.5);
	\coordinate(b) at (3,1);
	\coordinate(bprime) at (2,0);
	\coordinate(c) at (4.3,2.5);
	\coordinate(x) at (2,2);
	\coordinate(y) at (4,1.5);
	\draw  (aprime) node[label=above:$a'$, label=below:$\{L\}$]{} -- (a) node[label=above:$a$]{} -- (c) node[label=above:$c$]{}; 
	\draw (a) --(b) node[label=below:$b$, label={[name=label node]below:$\{L,R\}$}]{}; 
	\draw[dashed, color=gray] (b) -- (bprime) node[fill=gray,circle, inner sep=2.2pt, label=below:$b'$]{}; 
	\draw (b) -- (c);
	\draw (a) -- (y) node[label=right:$y$, label=below right:$\{L\}$]{} -- (b);
\end{scope}
\end{tikzpicture}
\hspace{2cm}
\begin{tikzpicture}
	\vertexnodes 
\begin{scope}[scale=1]
	\coordinate(aprime) at (0.8,2.9);
	\coordinate(a) at (1.7,2.5);
	\coordinate(b) at (3,1);
	\coordinate(bprime) at (2,0);
	\coordinate(c) at (4.3,2.5);
	\coordinate(x) at (2,2);
	\coordinate(y) at (4,1.5);
	\draw  (aprime) node[label=above:$a'$, label=below:$\{L\}$]{} -- (a) node[label=above:$a$]{} -- (c) node[label=above:$c$]{}; 
	\draw (a) --(b) node[label=below:$b$, label={[name=label node]below:$\{L,R\}$}]{}; 
	\draw[dashed, color=gray] (b) -- (bprime) node[fill=gray,circle, inner sep=2.2pt, label=below:$b'$]{}; 
	\draw (b) -- (c);
	\draw (c) -- (y) node[label=right:$y$, label=below right:$\{R\}$]{};
\end{scope}

\end{tikzpicture}
\begin{tikzpicture}
	\vertexnodes 
\begin{scope}[scale=1]
	\coordinate(aprime) at (0.8,2.9);
	\coordinate(a) at (1.7,2.5);
	\coordinate(b) at (3,1);
	\coordinate(bprime) at (2,0);
	\coordinate(c) at (4.3,2.5);
	\coordinate(x) at (2,2);
	\coordinate(y) at (4.3,1.5);
	\draw  (aprime) node[label=above:$a'$, label=below:$\{L\}$]{} -- (a) node[label=above:$a$]{} -- (c) node[label=above:$c$]{}; 
	\draw (a) --(b) node[label=below:$b$, label={[name=label node]below:$\{L,R\}$}]{}; 
	\draw[dashed, color=gray] (b) -- (bprime) node[fill=gray,circle, inner sep=2.2pt, label=below:$b'$]{}; 
	\draw (b) -- (c);
\draw (y) node[style=draw,shape=rectangle,rotate=48,fill=none]{only $\{L,R\}$};
\end{scope}

\end{tikzpicture}
\end{center}
\caption{Cases 1.1, 1.2, and 1.3 of the proof of Theorem~\ref{theo:main}}\label{fig:cases-one}
\end{figure}

	\emph{Case 1.} $\LL(b)=\{LR\}$ (see Figure~\ref{fig:cases-one}).

	\emph{Case 1.1.} There is a vertex $y\in V(G\setminus b')\setminus\{a,a',b,c\}$
	that comes after $b$ in $\bar v$ and $\LL(y)=\{L\}$.
	Then $G*c$ is isomorphic to $Q$.
	
	\emph{Case 1.2.} 
	Assume that
	there is a vertex $y\in V(G\setminus b')\setminus\{a,a',b,c\}$
	that comes after $b$ in $\bar v$ with 
	$\LL(y)=\{R\}$.
	Then the set $\{a,a',b,b',c,y\}$ induces a graph isomorphic to $N$ in $G$,
	a contradiction to the assumptions.

	\emph{Case 1.3.} 
	Every vertex $y\in V(G\setminus b')\setminus\{a,a',b,c\}$
	that comes after $b$ in $\bar v$ has label $\LL(y)=\{L,R\}$
	Then the interval $b,\ldots,c$ of $\bar v$ is $\LL$-constant,
	so by Lemma~\ref{lem:reorder-vertices} we can exchange the positions of $b$ and $c$ in $\bar v$,
	obtaining a new ordering $\bar w$ of $V(G\setminus b')$ witnessing that
	$(G\setminus b',(a,b),\bar w,\LL)$ is a thread block. But then $G$ is a thread graph:
	To see this, place $b'$ between $w_{n-1}$ and $w_n$
	and let $\LL(b'):=\{R\}$. 
	
	This contradicts $G$ being an obstruction.

\begin{figure}[h]
\begin{center}	
\begin{tikzpicture}
	\vertexnodes 
\begin{scope}[scale=1]
	\coordinate(aprime) at (0.8,2.9);
	\coordinate(a) at (1.7,2.5);
	\coordinate(b) at (3,1);
	\coordinate(bprime) at (2,0);
	\coordinate(c) at (4.3,2.5);
	\coordinate(x) at (2,1.5);
	\coordinate(y) at (4,1.5);
	\draw  (aprime) node[label=above:$a'$, label=below:$\{L\}$]{} -- (a) node[label=above:$a$]{} -- (c) node[label=above:$c$]{}; 
	\draw (a) --(b) node[label=above:$b$, label=below:$\{L\}$]{}; 
	\draw[dashed, color=gray] (b) -- (bprime) node[fill=gray,circle, inner sep=2.2pt, label=below:$b'$]{}; 
	\draw (b) -- (x) node[label=left:$x$, label={[name=label node]below left:$\{(L,)R\}$}]{} -- (c);
	\draw (c) -- (y) node[label=right:$y$, label=below right:$\{R\}$]{};
	\draw[xshift=-2cm,snake=bumps] (x) -- (a);
\end{scope}
\end{tikzpicture}
\hspace{2cm}
\begin{tikzpicture}
	\vertexnodes 
\begin{scope}[scale=1]
	\coordinate(aprime) at (0.8,2.9);
	\coordinate(a) at (1.7,2.5);
	\coordinate(b) at (3,1);
	\coordinate(bprime) at (2,0);
	\coordinate(c) at (4.3,2.5);
	\coordinate(x) at (2,1.5);
	\coordinate(y) at (4,1.5);
	\draw  (aprime) node[label=above:$a'$, label=below:$\{L\}$]{} -- (a) 
	node[label=above:$a$]{} -- (c) node[label=above:$c$]{}; 
	\draw (a) --(b) node[label=above:$b$, label=below:$\{L\}$]{}; 
	\draw[dashed, color=gray] (b) -- (bprime) node[fill=gray,circle, inner sep=2.2pt, label=below:$b'$]{}; 
	\draw (y) -- (a);
	\draw (y) -- (x);
	\draw (c) -- (y) node[label=right:$y$, label={[name=label node]below right:$\{L,R\}$}]{};
	\draw (b) -- (x) node[label=left:$x$, label={[name=label node]below left:$\{(L,)R\}$}]{} -- (c);
\draw[xshift=-2cm,snake=bumps] (x) -- (a);
\end{scope}

\end{tikzpicture}
\begin{tikzpicture}
	\vertexnodes 
\begin{scope}[scale=1]
	\coordinate(aprime) at (0.8,2.9);
	\coordinate(a) at (1.7,2.5);
	\coordinate(b) at (3,1);
	\coordinate(bprime) at (2,0);
	\coordinate(c) at (4.3,2.5);
	\coordinate(x) at (2,1.5);
	\coordinate(y) at (4.3,1.5);
	\draw  (aprime) node[label=above:$a'$, label=below:$\{L\}$]{} -- (a) node[label=above:$a$]{} -- (c) 
	node[label=above:$c$]{}; 
	\draw (a) --(b) node[label=above:$b$, label=below:$\{L\}$]{}; 
	\draw[dashed, color=gray] (b) -- (bprime) node[fill=gray,circle, inner sep=2.2pt, label=below:$b'$]{}; 
	\draw (y) node[style=draw,shape=rectangle,rotate=48,fill=none]{only $\{L\}$};
	\draw (b) -- (x) node[label=left:$x$, label={[name=label node]below left:$\{(L,)R\}$}]{} -- (c);
	\draw[xshift=-2cm,snake=bumps] (x) -- (a);
\end{scope}

\end{tikzpicture}
\end{center}
\caption{Cases 2.1, 2.2, and 2.3 of the proof of Theorem~\ref{theo:main}.
The spiral indicates that the edge may be present or not.}\label{fig:cases-two}
\end{figure}
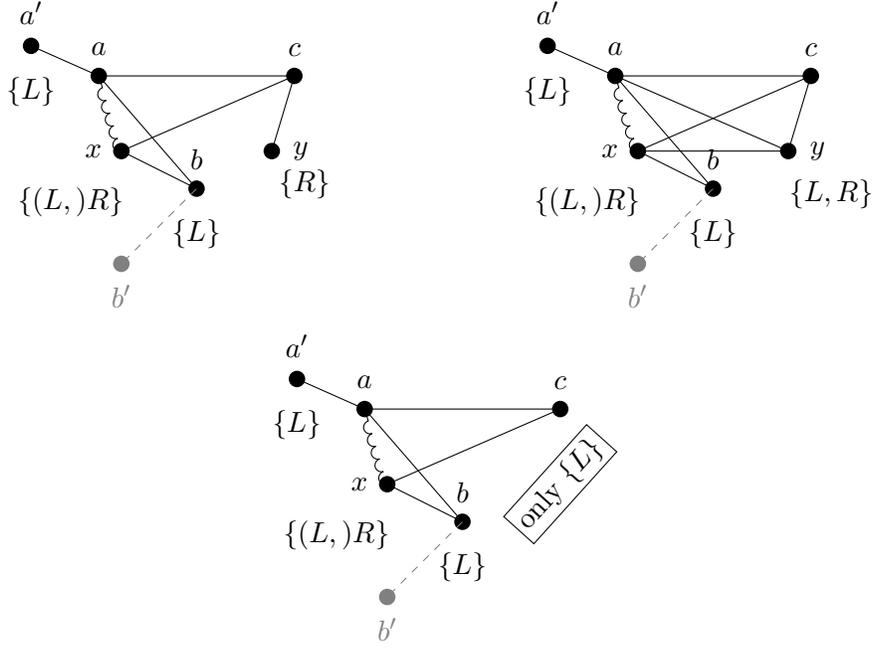

	\emph{Case 2.} $\LL(b)=\{L\}$ (see Figure~\ref{fig:cases-two}).
	Since $G_0$ is $2$-connected, there exists a vertex $x\in V(G\setminus b')\setminus\{a,a',b,c\}$
	that comes before $b$ in the ordering $\bar v$ and $R\in\LL(x)$.
	
	\emph{Case 2.1.} There is a vertex $y\in V(G\setminus b')\setminus\{a,a',b,c,x\}$
	that comes after $b$ in $\bar v$ and $\LL(y)=\{R\}$.
	If $L\notin\LL(x)$, then the set $\{a,b,b',c,x,y\}$ induces a graph isomorphic to $Q$ in
	$G\setminus a'$, a contradiction to $G$ being an
	obstruction. If $ L\in\LL(x)$, then $\{a,b,b',c,x,y\}$ induces a graph isomorphic to $N$ in
   $(G\setminus a')*x$, a contradiction to $G$ being an obstruction. 
	
	\emph{Case 2.2.} There is a vertex $y\in V(G\setminus b')\setminus\{a,a',b,c,x\}$
	that comes after $b$ in $\bar v$ and $\LL(y)=\{L,R\}$.
	If $L\notin\LL(x)$, then the set $\{a,b,b',c,x,y\}$ induces a graph isomorphic to 
	$N$ in $((G\setminus a')*c)*a$, a
	contradiction to $G$ being an obstruction.
   If $L\in\LL(x)$ then the set $\{a,b,b',c,x,y\}$ induces a graph isomorphic to $Q$ in 
	$(G\setminus a')*c$, a contradiction to $G$ being an obstruction.

	\emph{Case 2.3.} 
	Every vertex $y\in V(G\setminus b')\setminus\{a,a',b,c,x\}$
	that comes after $b$ in $\bar v$ has label $\LL(y)=\{L\}$.
	Then the interval $b,\ldots,c$ of $\bar v$ is $\LL$-constant,
	so by Lemma~\ref{lem:reorder-vertices} we can exchange the positions of $b$ and $c$ in $\bar v$,
	obtaining a new ordering $\bar w$ of $V(G\setminus b')$ witnessing that
	$(G\setminus b',(a,b),\bar w,\LL)$ is a thread block. But then $G$ is also a thread graph:
	To see this, place $b'$ between $w_{n-1}$ and $w_n$
	and let $\LL(b'):=\{R\}$. 
	
	This is a contradiction to $G$ being an obstruction.
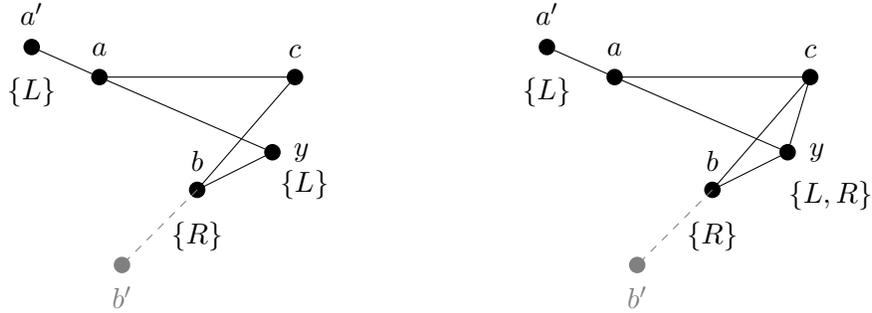
\begin{figure}[h]
\begin{center}	
\begin{tikzpicture}
	\vertexnodes 
\begin{scope}[scale=1]
	\coordinate(aprime) at (0.8,2.9);
	\coordinate(a) at (1.7,2.5);
	\coordinate(b) at (3,1);
	\coordinate(bprime) at (2,0);
	\coordinate(c) at (4.3,2.5);
	\coordinate(x) at (2,1.5);
	\coordinate(y) at (4,1.5);
	\draw  (aprime) node[label=above:$a'$, label=below:$\{L\}$]{} -- (a) node[label=above:$a$]{} -- (c) node[label=above:$c$]{}; 
	\draw (c) --(b) node[label=above:$b$, label=below:$\{R\}$]{}; 
	\draw[dashed, color=gray] (b) -- (bprime) node[fill=gray,circle, inner sep=2.2pt, label=below:$b'$]{}; 
	\draw (a) -- (y) node[label=right:$y$, label=below right:$\{L\}$]{} --(b);
\end{scope}
\end{tikzpicture}
\hspace{2cm}
\begin{tikzpicture}
	\vertexnodes 
\begin{scope}[scale=1]
	\coordinate(aprime) at (0.8,2.9);
	\coordinate(a) at (1.7,2.5);
	\coordinate(b) at (3,1);
	\coordinate(bprime) at (2,0);
	\coordinate(c) at (4.3,2.5);
	\coordinate(x) at (2,1.5);
	\coordinate(y) at (4,1.5);
	\draw  (aprime) node[label=above:$a'$, label=below:$\{L\}$]{} -- (a) 
	node[label=above:$a$]{} -- (c) node[label=above:$c$]{}; 
	\draw (c) --(b) node[label=above:$b$, label=below:$\{R\}$]{}; 
	\draw[dashed, color=gray] (b) -- (bprime) node[fill=gray,circle, inner sep=2.2pt, label=below:$b'$]{}; 
	\draw (y) -- (b);
	\draw (y) -- (a);
	\draw (c) -- (y) node[label=right:$y$, label={[name=label node]below right:$\{L,R\}$}]{};
\end{scope}
\end{tikzpicture}
\end{center}
\caption{Cases 3.1, and 3.2 of the proof of Theorem~\ref{theo:main}}\label{fig:cases-three}
\end{figure}

\medskip\noindent
	\emph{Case 3.} $\LL(b)=\{R\}$  (see Figure~\ref{fig:cases-three}).\\
	Since $G_0$ is $2$-connected, there exists a vertex 
	$y\in V(G\setminus b')\setminus\{a,a',b,c\}$
	that comes after $b$ in $\bar v$ and $L\in\LL(y)$. 

	\smallskip\noindent
	\emph{Case 3.1.} $\LL(y)=\{L\}$.\\
	Then $G$ is isomorphic to $Q$.
	
	\smallskip\noindent
	\emph{Case 3.2.} $\LL(y)=\{L,R\}$.\\
	Then $G*c$ is isomorphic to $N$, hence $G\sim N$, a contradiction to our assumptions.

	\smallskip
	Hence, if $G$ is an obstruction for $\CC$ such that $G\not\sim C_5$ and $G\not\sim N$,
	then only Cases 1.1.\ and 3.1.\ can occur and $G\sim Q$. 
\qed

\section{Conclusion}
The celebrated Robertson-Seymour Theorem shows the \emph{finiteness} of the obstruction sets 
for classes of graphs that are closed under taking minors. However, the cardinality
of such a set can be enormous.
While it is an open question, whether a similar theorem holds for classes of graphs that are closed
under taking vertex-minors, it is known that if the obstruction set has bounded rank-width, then
the obstruction set is finite. This implies that for every integer $k\geq 0$ the obstruction 
set for the class graphs of linear rank-width at most $k$ is finite. But until now, none of these
sets were known explicitely. 
In this paper, we have exhibited the finite set of minimal excluded vertex-minors for the class 
of linear rankwidth at most $1$. 
A natural next step would be to determine the obstruction set
for the graphs of linear rankwidth at most $2$. Nevertheless, we expect the number of
obstructions to be large. 
While there are two minimal excluded minors characterizing the class of graphs of path-width 
at most $1$, the class of graphs of path-width at most $2$ is characterized by 
$110$ minimal excluded minors~\cite{KinnersleyL94}.

\bibliographystyle{plain}
\newcommand{\hyphen}{-}
\bibliography{AFP-lrw1}

\begin{thebibliography}{10}

\bibitem{Assman81}
S.F. Assmann, G.W. Peck, Maciej~M. Sys{\l}o, and Jerzy {\.Z}ak.
\newblock The bandwidth of caterpillars with hairs of length 1 and 2.
\newblock {\em SIAM J. Algebraic Discrete Methods}, 2:387–393, 1981.

\bibitem{BandeltM86}
Hans-J{\"u}rgen Bandelt and Henry~Martyn Mulder.
\newblock Distance-hereditary graphs.
\newblock {\em J. Comb. Theory, Ser. B}, 41(2):182--208, 1986.

\bibitem{Bouchet94}
Andr{\'e} Bouchet.
\newblock Circle graph obstructions.
\newblock {\em J. Comb. Theory, Ser. B}, 60(1):107--144, 1994.

\bibitem{Chinn82}
Phyllis~Z. Chinn, Alexander~K. Dewdney, and Norman~E. Gibbs.
\newblock The bandwidth problem for graphs and matrices—a survey.
\newblock {\em Journal of Graph Theory}, 6:223--254, 1982.

\bibitem{Ganian10}
Robert Ganian.
\newblock Thread graphs, linear rank-width and their algorithmic applications.
\newblock In Costas~S. Iliopoulos and William~F. Smyth, editors, {\em IWOCA},
  volume 6460 of {\em Lecture Notes in Computer Science}, pages 38--42.
  Springer, 2010.

\bibitem{harary69}
Frank Harary.
\newblock {\em Graph Theory}.
\newblock Addison-Wesley, New York, 1969.

\bibitem{KinnersleyL94}
Nancy~G. Kinnersley and Michael~A. Langston.
\newblock Obstruction set isolation for the gate matrix layout problem.
\newblock {\em Discrete Applied Mathematics}, 54(2-3):169--213, 1994.

\bibitem{MonienS85}
Burkhard Monien and Ivan~Hal Sudborough.
\newblock Bandwidth constrained np-complete problems.
\newblock {\em Theor. Comput. Sci.}, 41:141--167, 1985.

\bibitem{Oum05}
Sang{\hyphen}il Oum.
\newblock Rank-width and vertex-minors.
\newblock {\em J. Comb. Theory, Ser. B}, 95(1):79--100, 2005.

\bibitem{Oum08}
Sang{\hyphen}il Oum.
\newblock Rank-width and well-quasi-ordering.
\newblock {\em SIAM J. Discrete Math.}, 22(2):666--682, 2008.

\bibitem{OumS06}
Sang{\hyphen}il Oum and Paul~D. Seymour.
\newblock Approximating clique-width and branch-width.
\newblock {\em J. Comb. Theory, Ser. B}, 96(4):514--528, 2006.

\bibitem{RobertsonS86}
Neil Robertson and Paul~D. Seymour.
\newblock Graph minors. {II}. {A}lgorithmic aspects of tree-width.
\newblock {\em J. Algorithms}, 7(3):309--322, 1986.

\end{thebibliography}
\end{document}